\documentclass[letterpaper, 10 pt, conference]{ieeeconf}

\usepackage{cite}
\usepackage{amsmath,amssymb,amsfonts}
\usepackage{graphicx}

\IEEEoverridecommandlockouts  
\usepackage{geometry}
\usepackage{xcolor}
\usepackage{balance}

\usepackage{tabularx} 
\usepackage{graphicx} 
\usepackage{cite} 
\usepackage[final]{hyperref} 
\usepackage{amsmath} 
\usepackage{amssymb}  
\usepackage{mathtools, cuted,bm}
\usepackage{etoolbox}
\patchcmd{\proof}{\indent}{}{}{}

\usepackage[justification=centering]{caption}

\newcommand{\twodots}{.\kern-0.1em.}
\newcommand{\myldots}{\kern-0.05em.\kern-0.01em.\kern-0.01em.\kern0.01em}

\usepackage{caption}
\usepackage{subcaption}
\usepackage{esvect}
\usepackage{lipsum, color}

\newcommand{\R}{\mathbb{R}}
\newcommand{\I}{\mathbb{I}}

\usepackage{relsize}
\usepackage{tikz}
\graphicspath{{../fig/}} 

\usepackage[margin=0.8cm]{caption}
\usepackage{pgfplots}
\pgfplotsset{compat=newest}
\usepackage{tabularx}
\pgfplotsset{plot coordinates/math parser=false}

\usepackage{ulem}
\usepackage{cancel}

\makeatletter
\def\munderbar#1{\underline{\sbox\tw@{$#1$}\dp\tw@\z@\box\tw@}}
\makeatother
\definecolor{olivegreen}{RGB}{128, 128, 0}

\usepackage{algpseudocode}
\usepackage{algorithm}

\newcommand\scalemath[2]{\scalebox{#1}{\mbox{\ensuremath{\displaystyle #2}}}}

\hypersetup{
	colorlinks=true,       
	linkcolor=black,        
	citecolor=black,        
	filecolor=magenta,     
	urlcolor=black         
}

\usepackage{algorithm}
\usepackage{algpseudocode}

\newtheorem{remark}{Remark}
\newtheorem{problem}{Problem}
\newtheorem{proposition}{Proposition}
\newtheorem{definition}{Definition}

\usepackage{comment}

\title{ \bf
Sample Efficient Certification of Discrete-Time Control Barrier Functions
}

\author{Sampath Kumar Mulagaleti, Andrea Del Prete  
\thanks{Financed by the European Union - Next Generation EU, Mission 4 Component 2 - CUP E53D23001130006.
S.K. Mulagaleti was with University of Trento, Italy, and A. Del Prete is with University of Trento, Italy. S.K.Mulagaleti is currently with IMT School for Advanced Studies Lucca, Italy. (Email: \url{s.mulagaleti@imtlucca.it},\url{andrea.delprete@unitn.it})}%
}

\begin{document}

\maketitle
\thispagestyle{empty}
\pagestyle{empty}

\begin{abstract}
Control Invariant (CI) sets are instrumental in certifying the safety of dynamical systems. Control Barrier Functions (CBFs) are effective tools to compute such sets, since the zero sublevel sets of CBFs are CI sets. However, computing CBFs generally involves addressing a complex robust optimization problem, which can be intractable. Scenario-based methods have been proposed to simplify this computation. Then, one needs to verify if the CBF actually satisfies the robust constraints. We present an approach to perform this verification that relies on Lipschitz arguments, and forms the basis of a certification algorithm designed for sample efficiency. Through a numerical example, we validated the efficiency of the proposed procedure.
\end{abstract}

\section{Introduction}
Guaranteeing reliable operation of dynamical systems is of paramount importance when deploying such systems in safety-critical applications. Safety is typically characterized by constraints on the system, such that guarantees can take the form of constraint-satisfaction certificates. These certificates can be provided using control invariant (CI) sets~\cite{Blanchini2015}, which are regions of the state space in which the system can be regulated indefinitely. Such sets can, in theory, be computed using a recursive procedure based on forward reachable sets~\cite{Bertsekas1972}. Unfortunately, a direct tractable application of this procedure is generally limited to linear settings~\cite{kerrigan2000robust}.
In the general setting, several approaches have been proposed to approximate the maximal CI set. Some approaches in this regard directly compute parameterized representations by enforcing the invariance conditions, e.g., using SOS programming~\cite{chesi2011domain}, DC programming~\cite{Fiacchini2010}, interval arithmetic~\cite{Bravo2005,Brown2023}, occupation measures~\cite{Korda2013}, zonotopes~\cite{Schäfer2024}, or a direct CI-set boundary approximation~\cite{VBOC}. Furthermore,~\cite{Mitchell2005} directly tackles the recursive procedure by formulating it using an optimal control framework. 
In recent years, Control Barrier Functions~\cite{Ames2019} (CBFs) have proven to be an effective way to characterize CI sets. They are Lyapunov-like functions that exhibit a bounded increase in CI level sets. For a review of recent developments, we refer to~\cite{Garg2024}. 

We focus on CBFs for discrete-time (DT) dynamical systems (DT-CBFs), which have gained recent interest following~\cite{agrawal2017discrete,Ahmadi2019}, with similar ideas having been presented in~\cite{Alberto2007}. While DT-CBFs have successfully been applied in many settings (e.g., model predictive control~\cite{Zeng2021}, safe reinforcement learning~\cite{Emam2022}) their synthesis remains an open problem. Some notable exceptions include the contributions in~\cite{Freire2023,wabersich2022predictive}. A standard approach to synthesize DT-CBFs involves enforcing inclusion constraints at finite samples. Then the question of certification arises, i.e., verifying generalizability of a candidate CBF.
This problem has been studied in the polynomial setting for continuous-time CBFs~\cite{Clark2021,Pond2023}. In~\cite{zhang2023exact}, an exact approach was developed for ReLU networks that exploits their piecewise affine structure. In \cite{Wang2023}, an SOS-based procedure was developed to verify CBFs for polynomial systems subject to input constraints. In the DT-CBF case, a recent approach in~\cite{chen2024verification} proposes to use a counter-example guided verification and synthesis procedure, in which finite termination guarantees are provided using a convex vector-valued CBF. Unfortunately, the procedure requires the solution of a global optimization problem at each iteration. A Branch-and-Bound algorithm to tackle this problem was developed in~\cite{Shakhesi2023}. An alternative to these approaches, of relevance to this paper, uses Lipschitz arguments.
In~\cite{Strong2024}, the conditions of~\cite{Alberto2007} were verified using a triangulation of the constraint set and leveraging Lipschitz arguments from the vertices. 

Our main contribution is a novel approach to verify DT-CBFs for general Lipschitz continuous systems. The central idea aims to enhance sampling efficiency by formalizing the intuition that the sampling density for verification can progressively decrease when moving from the boundary of the CI set to its interior. While the proposed idea is similar to the iteratively-refined grid-based approach in \cite{Tan2022}, we emphasize two important points. First, the approach of \cite{Tan2022} was developed for CT systems, while we deal with DT systems, which requires the development of new sampling density bounds. Second, our approach is complementary to~\cite{Tan2022}, in that their iterative algorithm can be used in lieu of our proposal as an automatic way to synthesize the sublevels. 
For completeness, we also discuss the synthesis problem, formulating it as a standard learning problem that also synthesizes a feasible control law. After a sample complexity analysis, we demonstrate the results on an example in which we exhibit the improved sample efficiency of our method.

\section{Problem Statement}
We consider discrete-time dynamical systems
\begin{align}
    \label{eq:system_OL}
    x^+ = \mathbf{f}(x,u)
\end{align}
where $\mathbf{f} : \R^{n_x} \times \R^{n_u} \to \R^{n_x}$ is Lipschitz continuous. The input is subject to constraints $u \in U$, and the state  to exclusion constraints $X :=\{x : x \notin X_{\mathrm{u}}\},$ where $X_{\mathrm{u}}$ is called an \textit{unsafe} set. We define the set of states from which we want to ensure safety as the \textit{safe} set $X_{\mathrm{s}} \subseteq X$. 
The set $X_{\mathrm{s}}$ could be trivially defined as the complement of $X_{\mathrm{u}}$, but this could make the safety certification problem unfeasible (there exist states in $X$ starting from which it is impossible not to reach $X_{\mathrm{u}}$), therefore, a gap between the two sets typically exists.

Our goal is to ensure that if $x_0 \in X_{\mathrm{s}}$, then $x_t \notin X_{\mathrm{u}}$ for all $ t\geq 0$.
We can achieve this goal by synthesizing a control invariant (CI) set that contains $X_{\mathrm{s}}$. 
A set $\mathbb{X}$ is called a CI set if and only if it satisfies the inclusions
\begin{align}
\label{eq:CI_conditions}
    \mathbb{X} \subseteq  X, &&
    \forall \ x \in \mathbb{X}, \ \exists \ u \in U \ : \ \mathbf{f}(x,u) \in \mathbb{X}.
\end{align}
Essentially, the conditions in~\eqref{eq:CI_conditions} state that if $x_0 \in \mathbb{X}$, then the subsequent states can be ensured to be safe, i.e., $x_t \notin X_{\mathrm{u}}$ with some $u_t \in U$. To ensure safety of $X_{\mathrm{s}}$ therefore we just need to enforce:
\begin{align}
\label{eq:safe_set}
    X_{\mathrm{s}} \subseteq \mathbb{X}.
\end{align}
Denoting $b$-sublevel sets of $h:\R^{n_x} \to \R$ as $S(b):=\{x:h(x) \leq b\}$, we parameterize CI sets as $0$-sublevel sets of a Lipschitz continuous function $h(x)$, i.e.,
\begin{align}
\label{eq:CI_parameterization}
    \mathbb{X}\leftarrow S(0):=\{x : h(x) \leq 0\}.
\end{align}
The inclusions $\mathbb{X} \subseteq X$ in~\eqref{eq:CI_conditions} and~\eqref{eq:safe_set} can be enforced by ensuring that
\begin{align}
\label{eq:safe_unsafe}
    h(x) > 0, \  \forall \ x \in X_{\mathrm{u}}, &&
    h(x) \leq 0, \ \forall \ x \in X_{\mathrm{s}}.
\end{align}
To ensure that the second inclusion in~\eqref{eq:CI_conditions} holds, we enforce that $h(x)$ is a control barrier function (CBF) adopting \cite[Def. 2]{Ahmadi2019}.

\begin{definition}
\label{defn_barrier}
    The function $h(x)$ is a CBF for System~\eqref{eq:system_OL} if there exists some $\alpha : \R \to \R$ satisfying
    \begin{align}
    \hspace{-5pt}
    \label{eq:alpha_properties}
    \alpha(y) \in [\min(0,y),\max(0,y)], 
    \end{align}
    and a set $D \subseteq \R^{n_x}$ satisfying the decay property
\begin{align}
    \label{eq:decay_property}
    \scalemath{0.99}{
    \exists u \in U : h(\mathbf{f}(x,u))-h(x) \leq -\alpha(h(x)), \forall x \in D.}
\end{align}
\end{definition}
\vspace{5pt}
In the following result, we demonstrate invariance of the set $\mathbb{X}$ parameterized as in~\eqref{eq:CI_parameterization} under CBF conditions.
\begin{proposition}
\label{prop:CBF}
    Suppose that the inequalities in~\eqref{eq:safe_unsafe} are satisfied, and the set $D$ (where~\eqref{eq:decay_property} holds) satisfies $\mathbb{X} \subseteq D.$
    Then, $\mathbb{X}$ is a CI set. Moreover, if $\alpha(y) \in (0,y]$ for $y>0$, then for any $x_0 \in D$, there exists a control sequence $\{u_t \in U, t\geq 0\}$ such that $\lim_{t \to \infty} x_t \in \mathbb{X}$.
\end{proposition}
\begin{proof}
    Take any $x \in \mathbb{X}$, such that $h(x) \leq 0$. From~\eqref{eq:alpha_properties}, we have $h(x) \leq \alpha(h(x)) \leq 0$, such that $h(x)-\alpha(h(x)) \leq 0$. Then, from $\mathbb{X} \subseteq D$ and~\eqref{eq:decay_property}, there exists some $u \in U$ verifying $h(\mathbf{f}(x,u)) \leq 0$, or equivalently, $\mathbf{f}(x,u) \in \mathbb{X}$ concluding the proof of the first part. For the second part, take any $x \in D\setminus \mathbb{X}$, such that $h(x)>0$. Since $\mathbb{X} \subseteq D$, we have $0 < \alpha(h(x)) \leq h(x)$, such that $h(x)-\alpha(h(x))<h(x)$. From~\eqref{eq:decay_property}, there exists a $u \in U$ such that $h(\mathbf{f}(x,u))<h(x)$.
    Telescoping this observation concludes the proof.
\end{proof}

In this paper, we parameterize 
\begin{align}
\label{eq:alpha_linear}
    \alpha(x)=\alpha x
\end{align}
for some $\alpha \in [0,1]$ for simplicity of presentation. Note that the results can be extended to more general functions verifying~\eqref{eq:alpha_properties}. Given this setup,
we state the problem of synthesizing CI sets as follows:
\begin{problem}
\label{eq:ideal_problem}
    Given System~\eqref{eq:system_OL}, sets $X_{\mathrm{s}}, X_{\mathrm{u}}, D \subseteq \R^{n_x}$ and $U \subseteq \R^{n_u}$ verifying $X_\mathrm{s} \subseteq X \subseteq D$, and scalar
    $\alpha \in [0,1]$, synthesize $h(x)$ that verifies~\eqref{eq:safe_unsafe} and~\eqref{eq:decay_property}. 
\end{problem}

While Problem~\ref{eq:ideal_problem} can be tackled directly for certain parameterizations of $\mathbf{f}$ and the sets defining the problem, it is generally intractable to enforce the inclusions in~\eqref{eq:safe_unsafe} and~\eqref{eq:decay_property}. Hence, a typical approach involves sampling from the sets, and enforcing the inclusions at these samples. Then, the question of verification arises, i.e., how can one ensure that a solution obtained with a sampling-based approach solves Problem~\ref{eq:ideal_problem}. In this spirit, we address the following problem in this paper.
\begin{problem}
\label{eq:sampled_problem} 
    (\textbf{Synthesis}) Given finite sample sets $Z_{\mathrm{s}}:=\{x_i \in X_{\mathrm{s}}, i \in \mathbb{I}_{\mathrm{s}}\}$, $Z_{\mathrm{u}}:=\{x_i \in X_{\mathrm{u}}, i \in \mathbb{I}_{\mathrm{u}}\}$ and $Z_{\mathrm{d}}:=\{x_i \in D, i \in \mathbb{I}_{\mathrm{d}}\}$ where $\mathbb{I}_{\mathrm{s}}$, $\mathbb{I}_{\mathrm{u}}$ and $\mathbb{I}_{\mathrm{d}}$ represent a finite collection of integers, the set $U$, and scalars 
    $\alpha \in [0,1]$ and $\delta \geq 0$, synthesize functions $h(x)$ and $u(x)$ such that the following inequalities hold:
    \begin{subequations}
    \label{eq:sampling_BF}
    \begin{align}
     \label{eq:sampling_BF:1}
      h(x) > 0, \ \forall \ x \in Z_{\mathrm{u}}, \quad h(x) \leq 0, \ \forall \ x \in Z_{\mathrm{s}}, \\
      \label{eq:sampling_BF:2}
      h(\mathbf{f}(x,u(x)))-(1-\alpha)h(x) \leq -\delta, \ \forall \ x \in Z_{\mathrm{d}}, \\
      \label{eq:sampling_BF:3}
      u(x) \in U, \ \forall \ x \in Z_{\mathrm{d}}
    \end{align}
    \end{subequations}
    (\textbf{Verification}) Derive conditions under which the synthesized $0$-sublevel set $S(0):=\{x : h(x) \leq 0\}$ solves Problem~\ref{eq:ideal_problem}, i.e., it is a CI set.
\end{problem}

In Problem~\ref{eq:sampled_problem}, we introduce a control law parameterization $u(x)$ to compute feasible control inputs, and a scalar $\delta \geq 0$ that will aid us in the verification. For synthesis, we propose a simple learning approach. For the verification problem, we present a sampling-based approach the exploits Lipschitz continuity, along with a barrier relaxation property in a manner similar to~\cite{boffi2020learning} to derive a sample-efficient certification methodology.
\section{Synthesis problem}
We formulate a simple learning problem to synthesize the functions $h(x)$ and $u(x)$ verifying~\eqref{eq:sampling_BF}. To this end, we parameterize the candidate barrier function as a feedforward neural network (FNN) with parameters $\theta_h \in \R^{\theta_h}$, and we denote $h(x) \leftarrow h(x;\theta_h)$. Then, assuming that the set $U$ is a polytope,
we parameterize the control law as the optimizer of the convex QP
\begin{align}
    \label{eq:u_parameterization}
    u(x;\theta_u):=\arg\min_{u \in U} \ \frac{1}{2} u^{\top} Q u + c(x;\theta_u)^{\top} u,
\end{align}
where $Q \succ 0$, and $c(x;\theta_u) : \R^{n_x} \to \R^{n_u}$ is an FNN with parameters $\theta_u \in \R^{n_{\theta_u}}$. 
For any $x \in \R^{n_x}$,~\eqref{eq:u_parameterization} is feasible since $\theta_u$ only enters through the objective. Furthermore, by the implicit function theorem, the optimizer is a differentiable function of $\theta_u$ under regularity assumptions on $U$, and continuity of the FNN $c(x;\theta_u)$ in $\theta_u$. 
Given this parameterization, the learning problem can be formulated as
\begin{subequations}
\label{eq:learning_problem_constraints}
\begin{align}
    &\min_{\theta_h,\theta_u} \ \|(\theta_h,\theta_u)\|_2^2 \\
    & \ \text{s.t.} \ \  h(x;\theta_h) \geq l_{\mathrm{u}}, \ \forall \ x \in Z_{\mathrm{u}}, \\ 
    & \ \qquad h(x;\theta_h) \leq 0, \ \forall \ x \in Z_{\mathrm{s}}, \\
    & \ \qquad \Delta h(x;\theta_h,\theta_u) \leq -\delta, \ \forall \ x \in Z_{\mathrm{d}},
\end{align}
\end{subequations}
where $l_{\mathrm{u}}>0$ is a user-defined scalar, and we denote
\begin{align*}
\scalemath{0.95}{
    \Delta h(x;\theta_h,\theta_u):=h(\mathbf{f}(x,u(x;\theta_u));\theta_h)-(1-\alpha)h(x;\theta_h).}
\end{align*}
Any nonlinear programming (NLP) solver~\cite{NoceWrig06} can be used to tackle Problem~\eqref{eq:learning_problem_constraints}. However, standard methods might face a severe computational challenge if the number of constraints is very large. One approach to tackle this issue involves penalizing constraint violation, and solving the resulting unconstrained problem. 
We denote $\bm{\theta}:=(\theta_h,\theta_u)$, and define the functions $\mathbf{h}_{\mathrm{u}}(x;\bm{\theta}) := \max\{l_{\mathrm{u}}-h(x;\theta_h),0\}$, $\mathbf{h}_{\mathrm{s}}(x;\bm{\theta}) := \max\{h(x;\theta_h),0\}$, and $\mathbf{h}_{\mathrm{d}}(x;\bm{\theta}):=\max\{\Delta h(x;\theta_h,\theta_u)+\delta,0\},$
using which we write the penalized problem as%
\begin{align}
    \label{eq:learning_problem_penalized}
    \min_{\bm{\theta}} \ \ &\tau_{\mathrm{r}} \|\bm{\theta}\|_2^2 + \frac{\tau_{\mathrm{u}}}{|\I_{\mathrm{u}}|} \sum_{x \in Z_{\mathrm{u}}} \mathbf{h}_{\mathrm{u}}(x;\bm{\theta}) + \\
     & \qquad \hspace{20pt} \frac{\tau_{\mathrm{s}}}{|\I_{\mathrm{s}}|} \sum_{x \in Z_{\mathrm{s}}} \mathbf{h}_{\mathrm{s}}(x;\bm{\theta}) + \frac{\tau_{\mathrm{d}}}{|\I_{\mathrm{d}}|} \sum_{x \in Z_{\mathrm{d}}} \mathbf{h}_{\mathrm{d}}(x;\bm{\theta}), \nonumber
\end{align}
where $\tau_{\mathrm{r}} \geq 0$ is the regularization parameter, and $\tau_{\mathrm{u}}, \tau_{\mathrm{s}}, \tau_{\mathrm{d}} >0$ are the penalty parameters. For high-enough penalty parameters, a stationary point of Problem~\eqref{eq:learning_problem_penalized} is feasible for Problem~\eqref{eq:learning_problem_constraints}. Furthermore, Problem~\eqref{eq:learning_problem_penalized} is separable and hence can be solved using stochastic gradient algorithms such as Adam~\cite{kingma2017adammethodstochasticoptimization}. Alternatively, solutions can be computed using quasi-Newton methods such as the L-BFGS algorithm~\cite{Byrd1995}.
We found it useful to warm-start a solver tackling Problem~\eqref{eq:learning_problem_penalized} with the solution of the classification problem derived by setting $\tau_d=0$ in Problem~\eqref{eq:learning_problem_penalized}.

\section{Verification problem}
In this section, we address the second part of Problem~\ref{eq:sampled_problem} for a candidate barrier function $h(x)$ obtained by solving the synthesis problem, e.g., Problem~\eqref{eq:learning_problem_constraints}. We assume that the candidate function satisfies the state constraints in~\eqref{eq:safe_unsafe}, such that the problem boils down to verifying~\eqref{eq:decay_property}. We remark that the satisfaction of the constraints in~\eqref{eq:safe_unsafe} can be checked using existing approaches, e.g., bound propagation \cite{gowal2019effectivenessintervalboundpropagation,fazlyab2023}, etc. These approaches, however, cannot check~\eqref{eq:decay_property} without further restrictions.
We also assume that the input is selected using a feasible control law, e.g.,~\eqref{eq:u_parameterization}, such that $u \in U$ is always guaranteed in the closed-loop system
\begin{align}
\label{eq:system}
    x^+=f(x):=\mathbf{f}(x,u(x)).
\end{align}
As per Definition~\ref{defn_barrier} and Proposition~\ref{prop:CBF}, a function $h(x)$ is a barrier function, and $S(0)$ is the corresponding CI set, if there exists some scalar $\bar{\alpha} \in [0,1]$ verifying
\begin{align}
\label{eq:condition_to_verify}
\hspace{-5pt}
    r(x):=h(f(x))-h(x) \leq -\bar{\alpha}h(x), \  \forall \ x \in S(0).
\end{align}
Note that the decay constant $\bar{\alpha}$ for verifying the barrier property does not necessarily have to equal $\alpha$ used for synthesizing $h(x)$.
We now present the framework under which we develop the verification methodology.

\subsection{Verification framework}
Suppose that a function $h(x)$ is obtained by solving the synthesis problem with a decay rate $\alpha \in [0,1]$. Let us define a sequence of scalars $\{\gamma_0,\cdots,\gamma_q\}$ verifying
\begin{align}
\label{eq:gamma_conditions}
    \gamma_0:=\min_{x \in S(0)} h(x), && \gamma_{i-1} < \gamma_{i}, && \gamma_q = 0,
\end{align}
for $i \in \{1,\cdots,q\}$. Then, define the sets
\begin{align}
\label{eq:segments}
    C_i:=\{x : h(x) \in [\gamma_{i-1},\gamma_i]\}, && i \in \{1,\cdots,q\},
\end{align}
and suppose that we have access to samples
\begin{align}
\label{eq:samples}
    \mathcal{D}_i:=\{x_j^i \in C_i \ : \ j \in \{1,\cdots,N_i\}\}
\end{align}
in each segment $C_i$ for $i\in\{1,\cdots,q\}$. Finally, assume that for some scalar $\delta \geq 0$, the condition
\begin{align}
    \label{eq:verified_condition}
    \hspace{-4pt}
    r(x) \leq -\alpha h(x) -\delta, \ \forall \ x \in \mathcal{D}_i, \ i \in \{1,\cdots,q\}
\end{align}
holds. Given~\eqref{eq:verified_condition}, we tackle the following problem as an alias for Part 2 of Problem~\ref{eq:sampled_problem}. As per the definition of $C_i$ in~\eqref{eq:segments}, the $0$-sublevel set of $h(x)$ is $S(0) = \cup_{i=1}^q C_i.$
\begin{problem}
    \label{main_problem}
    Derive conditions on the datasets $\mathcal{D}_i$ for $i \in \{1,\cdots,q\}$ such that for given $\bar{\alpha} \in [0,1]$,  verification of~\eqref{eq:verified_condition} implies verification of~\eqref{eq:condition_to_verify}. 
\end{problem}

Problem \ref{main_problem} requires us to derive sampling density bounds on datasets $\mathcal{D}_i$ such that the observation in \eqref{eq:verified_condition} can be extrapolated to conclude that \eqref{eq:condition_to_verify} holds.
\begin{proposition}
\label{prop:main_result}
    Suppose that for every $i \in \{1,\cdots,q\}$, the sample set $\mathcal{D}_i$ forms an $\epsilon_i$-net over $C_i$, i.e.,
    \begin{align}
        \label{eq:eps_net}
        \forall \ x \in C_i, \ \exists \ \bar{x} \in \mathcal{D}_i \ : \ \|x-\bar{x}\| \leq \epsilon_i.
    \end{align}
    Then, if~\eqref{eq:verified_condition} is satisfied, $h(x)$ is a barrier function for System~\eqref{eq:system} with decay rate $\bar{\alpha}\in [\alpha, 1]$, i.e., it satisfies
    \begin{align}
        \label{eq:relaxed_barrier}
        r(x) \leq -\bar{\alpha}h(x), && \forall x \in S(0)
    \end{align}
    if the sampling resolution parameter satisfies
    \begin{align}
        \label{eq:epsilon_condition}
        \epsilon_i \leq \zeta(\bar{\alpha},\delta, \gamma_i):=\cfrac{\delta + (\bar{\alpha} - \alpha) |\gamma_i|}{L_hL_f+(1-\bar{\alpha}) L_h}
    \end{align}
    for all $i \in \{1,\cdots,q\}$, where $L_h$ and $L_f$ upper-bound the Lipschitz constants of the functions $h(x)$ and $f(x)$.
\end{proposition}
\begin{proof}
    For a given $i \in \{1,\cdots,q\}$ and any $x \in C_i$, there exists some $\bar{x} \in \mathcal{D}_i$ verifying the inequalities
\begin{align}
        h(x^+)-(1-\bar{\alpha})h(x) & \label{eq:main_inequalities} \\
        &  \hspace{-80pt}  \leq  h(\bar{x}^+)+L_hL_f \epsilon_i-(1-\bar{\alpha})(h(\bar{x})-L_h \epsilon_i) \nonumber \\
        & \hspace{-80pt} = (h(\bar{x}^+)-h(\bar{x}))+\bar{\alpha}h(\bar{x})+(L_hL_f+(1-\bar{\alpha})L_h)\epsilon_i \nonumber \\
        & \hspace{-80pt} \leq (\bar{\alpha}-\alpha)h(\bar{x})-\delta + (L_hL_f+(1-\bar{\alpha})L_h)\epsilon_i \nonumber \\
        & \hspace{-80pt} \leq (\bar{\alpha}-\alpha)\gamma_i-\delta+ (L_hL_f+(1-\bar{\alpha})L_h)\epsilon_i \nonumber
           \end{align}
    where the first inequality follows from Lipschitz continuity of $h$ and $f$, the second inequality from~\eqref{eq:verified_condition}, and the final equality from the fact that $h(\bar{x})\leq \gamma_i$ since $\bar{x} \in \mathcal{D}_i$. Enforcing the right-hand-side to be non-positive concludes the proof.
\end{proof}

From Proposition~\ref{prop:main_result}, we observe that the choice of $\bar{\alpha}$ reflects the goal of the certification procedure. Selecting $\bar{\alpha}=1$ makes sense if our goal is to certify invariance of the set $S(0)$. This results in larger bounds on $\epsilon_i$, hence fewer samples are required. On the other hand, by choosing $\bar{\alpha}=\alpha$ we verify the barrier property with the original decay rate, resulting in smaller bounds on $\epsilon_i$, that are independent of $i$. Finally, we note that since $\{\gamma_0,\cdots,\gamma_q\}$ is a strictly increasing sequence, for any $x \in C_i$ and $i \in \{1,\cdots,q\}$, there always exists some $\bar{x} \in C_i$ such that $\|x-\bar{x}\| \in (0,\epsilon_i]$.

In the sequel, we present the sample complexity bounds to certify the $\epsilon_i$-net conditions.
Before we proceed, we present a brief discussion on the significance of the scalar $\delta$. In particular, we show that for $\delta>0$, a subset of $S(0)$ is an invariant set. We present the first result for the case when $h(x^+)-(1-\alpha)h(x)\leq -\delta$ holds for all $x \in S(0)$. Then, we present a new result when the decay condition is only verified at finite samples as in~\eqref{eq:verified_condition}, with the samples verifying~\eqref{eq:epsilon_condition}.

\begin{proposition}
\label{prop:delta_significance}
If $v(x):=r(x)+\alpha h(x) \leq -\delta$ holds for all $x \in S(0)$ for some $\delta>0$, then $S(-\delta/\alpha)$ is an invariant subset of the invariant set $S(0)$.
\end{proposition}
\begin{proof}
 For any $x_0 \in S(0)$ and $t \geq 0$, we have $ h(x_{t})\leq (1-\alpha)^t h(x_0)-\sum_{i=0}^{t-1}(1-\alpha)^i \delta.$
    For any $x_0$ such that $h(x_0)\leq -\delta/\alpha$, we have
    \begin{align}
        \label{eq:attractor_bound}
        h(x_t) \leq -\frac{\delta}{\alpha}\left((1-\alpha)^t+\sum_{i=0}^{t-1} \alpha(1-\alpha)^{i} \right).
    \end{align}
    For any $t\geq 1$, since $\alpha \in [0,1]$, we have $\sum_{i=0}^{t-1} \alpha(1-\alpha)^{i} = 1-(1-\alpha)^t$, such that $h(x_t)\leq -\delta/\alpha$.
\end{proof}

We now derive an analogous invariant set if $h(x)$ is certified following Proposition~\ref{prop:main_result}. 
\begin{proposition}
\label{prop:sequence_result}
  For any sequence $\{\gamma_0,\cdots,\gamma_q\}$ verifying~\eqref{eq:gamma_conditions} and $\bar{\alpha} \in [\alpha,1]$, if the sets $\mathcal{D}_i$ verify~\eqref{eq:eps_net} and~\eqref{eq:epsilon_condition}, and~\eqref{eq:verified_condition} is verified, then $S(\hat{\gamma})$ is an invariant subset of the invariant set $S(0)$, where
  \begin{align}
  \label{eq:constants_definition}
      \hat{\gamma}:=-\left(\cfrac{1-\bar{\alpha}}{\alpha(1-\bar{\alpha})+\bar{\alpha}L_f}\right) \delta.
  \end{align}
\end{proposition}
\begin{proof}
    For any initial state $x_0 \in S(\hat{\gamma})$, define the index $i_0 :=\min\{\gamma_i \in \{1,\cdots,q\} : h(x_0) \leq \gamma_i\}$, such that $x_0 \in C_{i_0}$. By assumption, there exists some $\bar{x} \in \mathcal{D}_{i_0}$ such that $\|x-\bar{x}\| \leq \epsilon_{i_0}$, with $\epsilon_{i_0}$ verifying~\eqref{eq:epsilon_condition} at $i=i_0.$
    Exploiting Lipschitz properties, we write the inequalities
    \begin{align}
        \label{eq:PI_inequalities}
        h(f(x_0)) &\leq h(f(\bar{x}))+L_h L_f \epsilon_{i_0} \\
        & \hspace{-30pt} \leq \underbrace{(1-\bar{\alpha})\left(\frac{L_hL_f+(1-\alpha)L_h}{L_hL_f+(1-\bar{\alpha})L_h}\right)}_{\mathbf{a}} \gamma_{i_0} + \nonumber \\
        & \hspace{50pt} + \underbrace{\left(\frac{(\bar{\alpha}-1)L_h}{L_hL_f+(1-\bar{\alpha})} \right)}_{\mathbf{b}}\delta \nonumber
    \end{align}
    where the first inequality follows from Lipschitz continuity and second inequality follows from~\eqref{eq:verified_condition}, \eqref{eq:epsilon_condition}, and the fact that $h(\bar{x}) \leq \gamma_{i_0}$. Thus, $\mathbf{a}\gamma_{i_0}+\mathbf{b}\delta$ upper bounds $h(x_1)$ for any $x_0 \in S(\gamma_{i_0})$. Repeating the procedure and taking worst case, we have $h(x_t) \leq \gamma_{i_t}$, where
    \begin{align}
    \label{eq:linear_system}
        \gamma_{i_{t+1}}=\mathbf{a}\gamma_{i_t}+\mathbf{b}\delta, && \forall t \geq 0.
    \end{align}
    The proof concludes trivially if $\bar{\alpha}=1$, since it implies $\gamma_{i_t}=0$ always, and $\hat{\gamma}=0$. For $\bar{\alpha} \in [\alpha,1)$, we have $\mathbf{a} \in (0,1)$ since $\mathbf{a}$ is strictly decreasing in $\bar{\alpha}$, with a largest value of $1-\alpha<1$.
    Hence,  System~\eqref{eq:linear_system} generates a monotonically increasing sequence which converges at $\hat{\gamma}$, i.e., $\hat{\gamma}=\mathbf{a}\hat{\gamma}+\mathbf{b}\delta$, concluding the proof.
\end{proof}
This result implies that if our main concern is invariance certification, it is in fact sufficient to certify that
\begin{align}
    r(x) \leq -\bar{\alpha} h(x), && \forall \ x \in S(\hat{\gamma})
\end{align}
for $\hat{\gamma}<0$ defined in~\eqref{eq:constants_definition}, or equivalently setting $\gamma_q = \hat{\gamma}$. Note that $-\delta/\alpha \leq \hat{\gamma}$ always holds, which implies that if verification is performed only at finite samples, then the CI set reaches closer to the boundary of $S(0)$.
\remark{Observe from Proposition \ref{prop:sequence_result} that the invariant set $S(\hat{\gamma})$ is nonempty only if $\delta$ is chosen small enough such that $\hat{\gamma}>\gamma_0$. Thus, the choice of $\delta$ depends on the parameters $\alpha$ and $\bar{\alpha}$, along with the Lipschitz constant $L_f$ of the closed-loop system.}
\subsection{Sample complexity bounds}
We now derive lower bounds on the number of samples required to verify the $\epsilon_i$-net conditions derived in Proposition \ref{prop:main_result} in a probabilistic manner. To this end, we first note that a straightforward method to verify the $\epsilon$-net conditions involves sampling from a simpler set $B_i \supseteq C_i$, and then perform rejection sampling to construct $\mathcal{D}_i$.
\begin{proposition}
\label{prop:box_samples}
    Suppose $B:=[\underline{x},\overline{x}]$, and let $\mathcal{D}:=\{x^i \in B:i \in \{1,\cdots,N\}\}$ consist of axis-aligned samples from $B$. The set $\mathcal{D}$ forms an $\epsilon$-net over $B$ if in each dimension, the grid resolution is $d:=2\epsilon/\sqrt{n}$,
    such that the total number of samples is $N=\prod_{k=1}^n \lceil \frac{\overline{x}_k-\underline{x}_k}{d}\rceil$.
\end{proposition}

While Proposition~\ref{prop:box_samples} is useful in exactly sampling for providing verification certificates, it might result in a very large number of samples depending on the volume of the sets $C_i$. Furthermore, in many applications, it might be sufficient to provide probabilistic certificates, i.e., bounds on the number of samples required to verify~\eqref{eq:epsilon_condition} with a certain probability, resulting in a probabilitic invariance certification of $S(0)$. Denoting $v(x):=r(x)+\alpha h(x)$, we
define the sets containing the states violating the barrier condition as
\begin{align}
\label{eq:violation set}
    \Omega_i:=\{x \in C_i : v(x)>-\delta\}, \ i=\{1,\cdots,q\}.
\end{align}
\begin{proposition}
\label{prop:omega_set}
    Define the largest ball in the set $\Omega_i$ as
    \begin{align}
    \label{eq:largest_ball}
    \hat{\epsilon}_i:=\max_{\epsilon} \{ \epsilon \geq 0 \ : \ \exists \ x \in \Omega_i \ : \ x \oplus \epsilon \mathcal{B} \subseteq \Omega_i\}
\end{align}
for each $i \in \{1,\cdots,q\}$ where $\mathcal{B}$ is the $2$-norm ball in $\R^{n_x}$, and define the complementary set as $C_i^{\mathrm{g}}:=C_i \setminus \Omega_i$. Then, $(i)$ For all $x \in \Omega_i$, there exists $\bar{\rho}>0$ such that for all $\rho \in (0,\bar{\rho})$, there exists some $\bar{x} \in C_i^{\mathrm{g}}$ verifying $\|x-\bar{x}\| \leq \hat{\epsilon}_i + \rho$; $(ii)$ If for a given $\bar{\alpha} \in [\alpha, 1]$, the inequality $\hat{\epsilon_i} \leq \zeta(\bar{\alpha},\delta, \gamma_i)$
holds for all $i \in \{1,\cdots,q\}$, then~\eqref{eq:relaxed_barrier} is verified.
\end{proposition}
\begin{proof}
    Part $(i)$ follows from the definition of $\hat{\epsilon}_i$ and the fact that $\Omega_i$ is an open set. Part $(ii)$ follows after observing that $v(x) \leq -\delta$ for all $x \in C^{\mathrm{g}}_i$, followed by using the same arguments as Proposition~\ref{prop:main_result} and exploiting continuity of the inequalities in~\eqref{eq:main_inequalities} as $\rho \to 0$.
\end{proof}
We now derive confidence bounds on $\hat{\epsilon}_i$ based on the number of samples $N_i$ composing the sets $\mathcal{D}_i \subset C_i$. 
Suppose that each set $C_i$ is endowed with a uniform probability distribution $\mathcal{U}_i$ with probability measure $\mu_i(\cdot)$. The probability of sampling a state $x \in C_i$ such that $v(x)>-\delta$, denoted as $\mathrm{Pr}_{x \sim \mathcal{U}_i}[v(x)>-\delta]$ is given by the $0$-$1$ risk $\mathrm{e}_i(h):=\mathbb{E}_{x \sim \mathcal{U}_i} [\mathbf{1}(v(x)>-\delta)].$
From the definition of $\Omega_i$, this $0$-$1$ risk is the measure of the set $\Omega_i$, i.e., $\mathrm{e}_i(h)=\mu_i(\Omega_i).$
From~\eqref{eq:verified_condition}, we see that the inequality $v(x)\leq -\delta$ holds over all samples $x \in \mathcal{D}_i$. This implies that empirical risk
\begin{align}
    \label{eq:empirical_risk}
    \hat{\mathrm{e}}_i(h):=\frac{1}{N_i} \sum_{j=1}^{N_i} \mathbf{1}(v(x_j^i)>-\delta) = 0.
\end{align}
\begin{proposition}
\label{prop:main_prob_res}
   Suppose the samples in $\mathcal{D}_i$ are i.i.d. For any $\theta \in (0,1)$, with a confidence of at least $1-\theta$, the $0$-$1$ risk of sampling $x \in \Omega_i$ is upper-bounded as
    \begin{align}
        \label{eq:EBI_res}
        \mathrm{e}_i(h) \leq \kappa_{\theta}(N_i):= \frac{7 \ln(2/\theta)}{3(N_i-1)}. \\ \nonumber
    \end{align}
\end{proposition}
\begin{proof}
    The proof follows from~\cite[Theorem 4]{maurer2009empirical} after observing that~\eqref{eq:empirical_risk} implies empirical mean and variance of the random variable $\mathbf{1}(v(x)>-\delta)$ equal $0$.
\end{proof}

Proposition~\ref{prop:main_prob_res} states that, given the evidence of samples in the set $\mathcal{D}_i$, the probability of encountering some $x \in \Omega_i$ where $v(x)>-\delta$ is upper-bounded by $\kappa_{\theta}(N_i)$. 
Since $\mathrm{e}_i(h)=\mu_i(\Omega_i)$, we have
\begin{align}
\label{eq:vol_ratio_bound}
    \mu_i(\Omega_i) \leq \kappa_{\theta}(N_i)
\end{align}
with a confidence of at least $1-\theta$. Since $\mu_i(\cdot)$ is the measure of the uniform distribution $\mathcal{U}_i$ over $C_i$, $\mu_i(\Omega_i)$ is the ratio of the Lebesgue measures $\mu_{\mathrm{Leb}}(\cdot)$ of the sets $\Omega_i$ and $C_i$, such that $\mu_{\mathrm{Leb}}(\Omega_i) \leq \kappa_{\theta}(N_i) \mu_{\mathrm{Leb}}(C_i)$
with a confidence of at least $1-\theta$. 
The following result from~\cite[Lemma C.1]{boffi2020learning} bounds the largest ball that can be included is subsets of measure less that $\kappa$ inside $C_i$.
\begin{proposition}
\label{prop:ball_size}
Suppose $C_i$ is full-dimensional, and
\begin{align}
\label{eq:ball_problem}
        \tilde{\epsilon}_i(\kappa):=\max\limits_{Z,x,x \oplus \epsilon \mathcal{B} \subseteq Z} \epsilon \ \ \text{s.t.} \ \ Z \subseteq C_i, \ \mu_i(Z) \leq \kappa.
    \end{align}
    Then, the size of the largest ball is bounded as
    \begin{align}
        \tilde{\epsilon}_i(\kappa) \leq \left(\frac{\kappa \mu_{\mathrm{Leb}}(C_i)}{\mu_{\mathrm{Leb}}(\mathcal{B})}\right)^{1/n_x}. \\ \nonumber
    \end{align}
\end{proposition} 

Since we know that the measure of $\Omega_i$ is upper-bounded as in~\eqref{eq:vol_ratio_bound}, it is feasible solution for Problem~\eqref{eq:ball_problem}. Hence, we conclude from Proposition~\ref{prop:ball_size} that $\hat{\epsilon}_i$ as defined in~\eqref{eq:largest_ball} is upper-bounded as
\begin{align}
\label{eq:epsilon_hat_bound}
    \hat{\epsilon}_i \leq \left(\frac{\kappa_{\theta}(N_i) \mu_{\mathrm{Leb}}(C_i)}{\mu_{\mathrm{Leb}}(\mathcal{B})}\right)^{1/n_x}
\end{align}
with a confidence at least $1-\theta$. We now exploit this property to derive number of samples required to probabilistically certify $h(x)$ to be a barrier function.
\begin{proposition}
\label{prop:probabilistic_certification}
    For a given $\bar{\alpha} \in [\alpha,1]$ and any $\theta \in (0,1)$, suppose that the number of i.i.d. samples $N_i$ composing $\mathcal{D}_i \subset C_i$ for each $i \in \{1,\cdots,q\}$ satisfies
    \begin{align}
    \label{eq:sample_bound_prob}
    \left(\frac{\kappa_{\theta}(N_i) \mu_{\mathrm{Leb}}(C_i)}{\mu_{\mathrm{Leb}}(\mathcal{B})}\right)^{1/n_x} \leq \zeta(\bar{\alpha},\delta, \gamma_i)
    \end{align}
    where $\kappa_{\theta}(N_i)$ in~\eqref{eq:EBI_res}, then $h(x)$ satisfies the barrier property in~\eqref{eq:relaxed_barrier} with a confidence of atleast $(1-\theta)^q$. 
\end{proposition}
\begin{proof}
    The proof follows from Proposition~\ref{prop:omega_set}, and observing that in each segment $C_i$, the bound in~\eqref{eq:epsilon_hat_bound} holds with a confidence of at least $1-\theta$. 
\end{proof}
From the definition of $\kappa_{\theta}(N_i)$ in \eqref{eq:EBI_res}, we see that \eqref{eq:sample_bound_prob} enforces a lower-bound on the number of iid samples composing the sets $\mathcal{D}_i$. Unfortunately, this bound increases exponentially in the state dimension $n_x$. Hence, future work can focus on deriving tighter bounds, which follows from the $0$-$1$ risk in Proposition \ref{prop:main_prob_res}.
\subsection{Certification algorithm}
\label{sec:certification}
In Algorithm~\ref{alg:cap}, we define a methodology to certify barrier functions. In Step 1, 
we estimate Lipschitz constants $L_h$ and $L_f$ of $h(x)$ and $f(x)$ respectively.
Since $L_h$ is typically given in an explicit functional form, e.g., a feedforward neural network, several established techniques can be used~\cite{szegedy2014,fazlyab2023}. Estimating $L_f$ however depends on the form of System~\eqref{eq:system}. Since we deal with closed-loop systems of the form $x^+ = \mathbf{f}(x,u(x))$ with $u(x)$ being an optimization-based controller, a simple approach is to sample a large number of pairs $\{x^i_1 \in S(0),x^i_2 \in S(0)\}$ for $i \in \{1,\cdots,N_f\}$, and estimate a lower bound to $L_f$ as
\begin{align*}
L_f > \max_{i \in \{1,\cdots,N_f\}} \frac{\|\mathbf{f}(x_1^i,u(x_1^i))-\mathbf{f}(x_2^i,u(x_2^i))\|}{\|x_1^i-x_2^i\|}.
\end{align*}
Alternatively, if Lipschitz constant of $u(x)$ can be computed, then a composition of the controller with the open-loop dynamics can also be used.
The study of methods to estimate these constants is beyond the scope of the current contribution.

Then, in Step 2, we select parameters $\{\gamma_0,\cdots,\gamma_q\}$ verifying~\eqref{eq:gamma_conditions}. To this end, we recall from Proposition~\ref{prop:sequence_result} that selecting $\gamma_q=\hat{\gamma}$ as defined in~\eqref{eq:constants_definition} is sufficient to verify invariance since $S(\hat{\gamma})$ is an invariant set. To select this sequence, the ideal problem to solve is
\begin{align}
    \label{eq:ideal_problem_samples}
    \min_{\gamma_i, i \in \{1,\cdots,q\}} \sum_{i=1}^{q} N_i  \ \text{s.t.} \ \eqref{eq:sample_bound_prob}, \gamma_{i-1}<\gamma_i, \ \gamma_q = \hat{\gamma},
\end{align}
where $N_i$ is the number of samples needed for an eps-ball cover. Unfortunately, solving Problem~\eqref{eq:ideal_problem_samples} exactly is generally intractable. While a feasible solution can be computed by selecting uniformly-spaced $\{\gamma_0,\cdots,\gamma_q\}$, we found that the sequence generated by the system $\gamma_{i+1}=\mathbf{a}\gamma_i+\mathbf{b}\delta$ defined in~\eqref{eq:linear_system} results in a fewer number of samples for verification. Following the selection of a sequence $\{\gamma_0,\cdots,\gamma_q\}$, we can select the sequence $\{\epsilon_1,\cdots,\epsilon_q\}$ to be equal to the upper-bounds in~\eqref{eq:epsilon_condition}.
Following the computation of this sequence, we sample points $\mathcal{D}_i$ verifying this $\epsilon_i$-net condition, and verify~\eqref{eq:verified_condition} at each of the points. Typically, this is a computationally expensive step that can be accelerated on parallelized hardware. Note that if Algorithm \ref{alg:cap} returns failure, it does so after encountering a point at which \eqref{eq:verified_condition} is not verified. Hence, this point can be appended to the training dataset to learn a new candidate barrier function, after warm-starting the solver at the previous solution. 


\begin{algorithm}[t]
\caption{Barrier function certifier}\label{alg:cap}
\begin{algorithmic}
\Require $f(x)$, $h(x)$, $\alpha \in [0,1]$, $\bar{\alpha} \in [\alpha,1]$, $\delta \geq 0$, $\theta \in (0,1)$
\State 1. Estimate $L_h$ and $L_f$
\State 2. Select $\{\gamma_0,\cdots,\gamma_q\}$
\State 3. Compute $\{\epsilon_1,\cdots,\epsilon_q\}$ verifying~\eqref{eq:epsilon_condition}
\State 4. Build sets $\mathcal{D}_i$ in~\eqref{eq:samples} that verify the $\epsilon_i$-bounds
\If{\eqref{eq:verified_condition} is verified for each $i \in \{1,\cdots,q\}$}
    \State Return \textbf{Certification success}
\Else
    \State Return \textbf{Certification fail}
\EndIf
\end{algorithmic}
\label{algorithm}
\end{algorithm}

\section{Numerical results}
We present a numerical example to demonstrate the efficacy of our CBF synthesis and verification schemes, in which we solve Problem~\eqref{eq:learning_problem_penalized} 
using the L-BFGS-B solver~\cite{Byrd1995} in \url{jaxopt} warm-stated by the ADAM solver \cite{kingma2017adammethodstochasticoptimization}, with the QP in~\eqref{eq:u_parameterization} solved using the 
closed-form solution in the case of diagonal matrix $Q$, given by $u(x;\theta_u):=\min(\max(-Q^{-1}c(x;\theta_u),\munderbar{u}),\bar{u}).$
\footnote{The code for reproducing this example can be found on \url{https://github.com/samku/DT-CBF-verification}. Simulations were performed on an Ubuntu 24.04 laptop equipped with Intel i9-14900HX processor, $32$GB of RAM, and Nvidia RTX4080 GPU running Cuda 12.}
We consider the system
\begin{align*}
    \begin{bmatrix} \dot{x}_1 \\ \dot{x}_2 \end{bmatrix} = \begin{bmatrix} x_2+\mathrm{cos}(x_1) \\ (1-x_1^2)x_2-x_1+\mathrm{sin}(x_1)+u \end{bmatrix}
\end{align*}
discretized using the Runge-Kutta 4 scheme with a timestep of $0.1$s. The system is subject to input constraints $U=[-2,2]$, unsafe set $X_{\mathrm{u}}=0.4\mathcal{B} \cup \{x : x^{\top} x \geq 2.8^2\}$ and safe set $X_{\mathrm{s}}=2\mathcal{B} \setminus 1.2\mathcal{B}$.

\begin{figure}[t]
    \centering
    \includegraphics[width=1\linewidth, trim=0.cm 0.cm 0.cm 0.cm, clip]{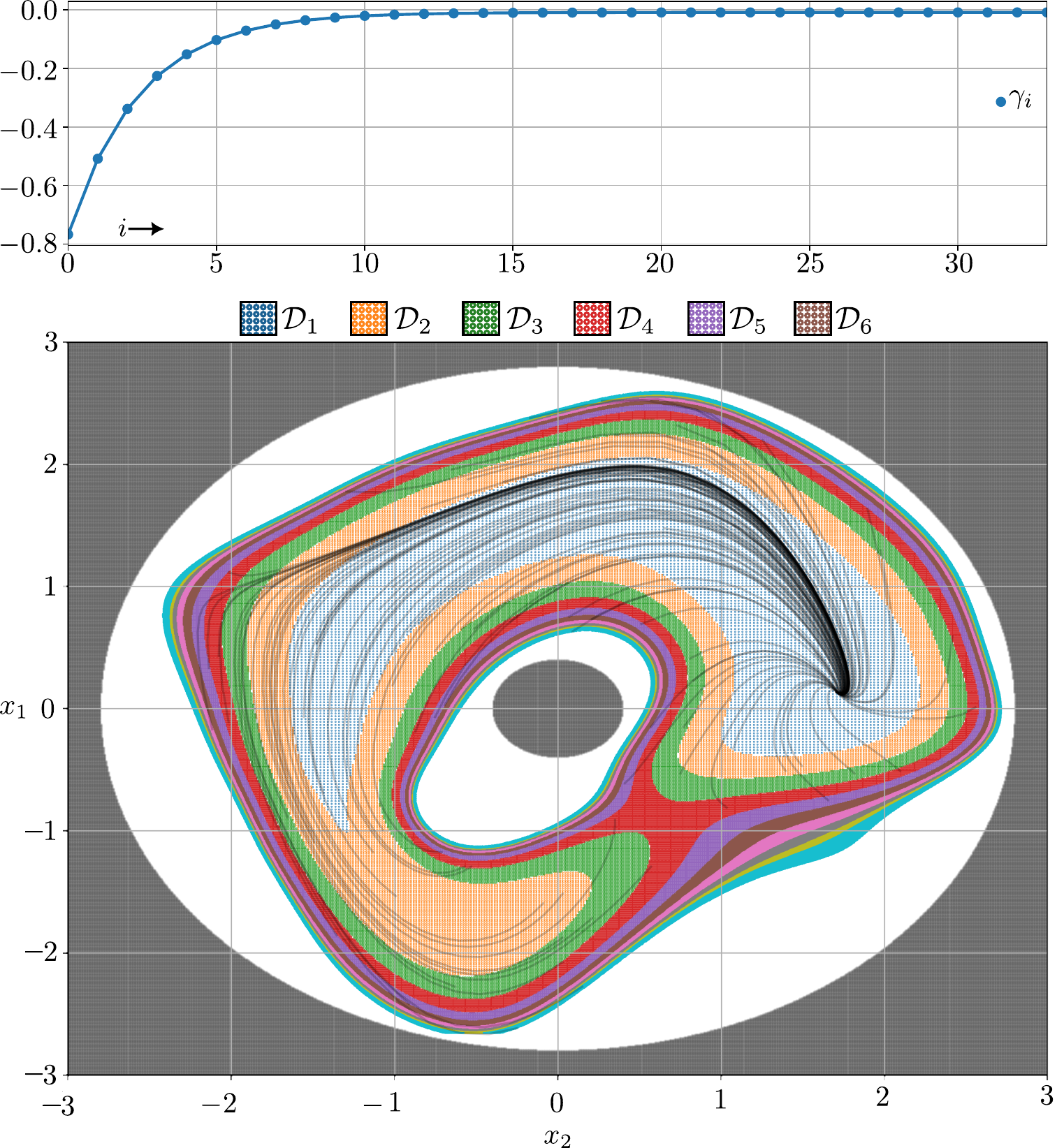}
    \captionsetup{width=\linewidth,justification=justified, singlelinecheck=false}
    \caption{(Top) Sequence $\gamma_i$ generated by~\eqref{eq:linear_system} for $\bar{\alpha}=0.4$; (Bottom) Corresponding sampled-sets $\mathcal{D}_i$, with some colors indicated in the legend. Also shown are the closed-loop trajectories demonstrating invariance. The gray region is $X_{\mathrm{u}}$. }
    \label{fig:Prob_1_fig_2}
\end{figure}

\textbf{Synthesis:}
We solve Problem~\eqref{eq:learning_problem_penalized} using $20000$ total samples from the set $3\mathcal{B}_{\infty}^2$. We parameterize the FNN $h(x;\theta_{\mathrm{h}})$ as a $2$-layer network with $10$ $\mathrm{tanh}$ activation units each, and $c(x;\theta_{\mathrm{u}})$ as a single layer FNN with $10$ $\mathrm{tanh}$ activation units. We select $D:=2.5\mathcal{B}_2^2$ in~\eqref{eq:decay_property}, and formulate Problem~\eqref{eq:learning_problem_penalized} with CBF parameters $\alpha=0.01$, $l_{\mathrm{u}}=0.1$ and $\delta=0.01$, penalty parameters $\tau_{\mathrm{u}}=2$, $\tau_{\mathrm{s}}=1$ and $\tau_{\mathrm{d}}=10$, and regularization parameter $\tau_{\mathrm{r}}=0.00025$. The results of this synthesis problem are plotted in Figure~\ref{fig:Prob_1_fig_2}(Bottom), in which we show the set $S(0)$, along with trajectories of the closed-loop system $x^+ = \mathbf{f}(x,u(x))$ initialized at several points in $S(0)$. We also show the unsafe region $X_{\mathrm{u}}$ demonstrating that the computed set satisfies $S(0) \subset X$.

\textbf{Verification:} To verify the computed CBF, we select $\alpha=0.2$, and follow the steps in Algorithm~\ref{algorithm}. We compute $\gamma_0=-1.003$ by solving $\min_{x \in S(0)} h(x)$. Then, we estimate $L_f = 1.4325$ and $L_h=1.6854$. For a given $\bar{\alpha} \in [\alpha,1)$, we first compute $\hat{\gamma}$ exploiting Proposition \ref{prop:sequence_result}, that we use to define $\gamma_q$. As a benchmark, we set $q=1$, which is equivalent to verifying the whole set $S(\hat{\gamma})$ in one shot. We denote the number of samples required to perform this verification as $N_{\mathrm{base}}$. 

To demonstrate the efficacy of our approach which involves dividing $S(\hat{\gamma})$ into segments defined by $\{\gamma_0,\cdots,\gamma_q\}$, we utilize two strategies to define the sequence. Firstly, we utilize the sequence generated by \eqref{eq:linear_system}. In Figure \ref{fig:Prob_1_fig_2}(Top), we plot the sequence generated with $\bar{\alpha}=0.4$, and in Figure \ref{fig:Prob_1_fig_2}, we show the sampled sets $\mathcal{D}_i$ used for verification. In Table \ref{tab:data_summary}, we report the values of $\hat{\gamma}$, along with the number of segments $q$ and total number of sample $N_{\mathrm{tot}} = \sum_{i=1}^{q} N_i$. The observation that $N_{\mathrm{tot}} < N_{\mathrm{base}}$ validates our approach. We also report the time required for verification, in which we observe that as the number of segments increase, the verification time increases. This is attributed to the fact that we perform rejection sampling, which requires sampling a much larger number of samples than those that contribute to the verification. Future work can focus on effective strategies to reduce this time.
\begin{table}[ht]
\centering
\begin{tabular}{c c c c c c}
\hline
$\bar{\alpha}$ & $N_{\mathrm{base}}$  & $q$ & $\hat{\gamma}$ & $N_{\mathrm{tot}}$ & Time [s] \\
\hline
$0.4$ & $2788210$ & $32$ & $-0.0086$ & $273676$ & $15.0716$ \\
$0.6$ & $2294426$  & $20$ & $-0.00425$ & $139697$ & $9.7489$ \\
$0.8$ & $2064837$  & $11$ & $-0.00168$ & $111581$ & $5.7082$ \\
\hline
\end{tabular}
\captionsetup{width=\linewidth,justification=justified, singlelinecheck=false}
\caption{Verification requirements using \eqref{eq:linear_system} to generate $\{\gamma_0,\cdots,\gamma_q\}$.}
\label{tab:data_summary}
\end{table}

As a second approach, we fix the value of $q$, and select uniformly spaced $\{\gamma_0,\cdots,\gamma_q\}$. In Figure \ref{fig:Prob_1_fig_1}, we plot the variation in number of samples required for different values of $q$ and $\bar{\alpha}$. We observe a large reduction in the number of samples required, validating our approach further. We report that consistently, the choice of $\gamma$ resulting from \eqref{eq:linear_system} results in a reduced total number of samples. However, as reported earlier, this would involve being unable to control the number of segments.

\begin{figure}
    \centering
    \includegraphics[width=1\linewidth, trim=0.cm 0.cm 0.cm 0.cm, clip]{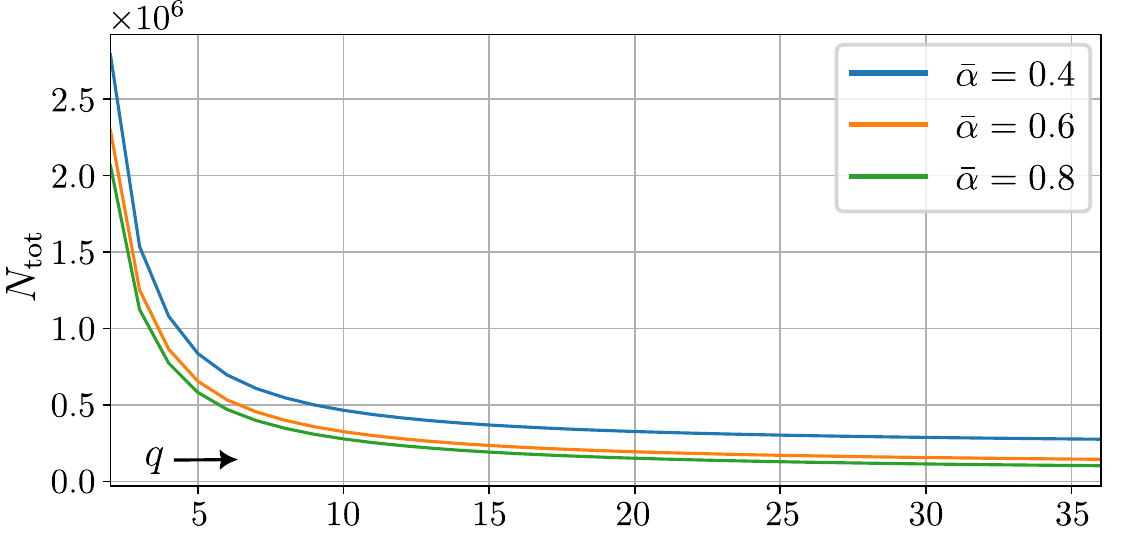}
    \captionsetup{width=\linewidth,justification=justified, singlelinecheck=false}
    \caption{Variation of number of samples required for certification with number of segments $q$, with $\gamma_i$ selected uniformly between $\gamma_0$ and $\hat{\gamma}$.}
    \label{fig:Prob_1_fig_1}
\end{figure}

\section{Conclusions}
In this paper, we have presented an approach to verify DT-CBFs by relying on Lipschitz arguments, with particular focus on certifying invariance of the $0$-sublevel set. To this end, we present a new result that characterizes the sampling density required for certification as a function of the level of the CBFs. This result formalizes the intuition that the sampling density can monotonically decrease as we move further from the boundary of the set. We have formulated a certification algorithm based on this result, that we demonstrate to be more sample-efficient that a one-shot approach. Future work can focus on the development of a synthesis procedure to compute certified-by-construction CBFs, and extending the methodologies to uncertain systems with parameterized CBFs.

\bibliographystyle{plain}
\bibliography{ieee_references}



\end{document}